\documentclass[runningheads]{llncs}
\usepackage[title]{appendix}
\usepackage{graphicx}
\usepackage{color}
 \usepackage{amssymb}
 \usepackage{subfig}
\usepackage{caption}
\usepackage{multirow}
\usepackage{stackengine}
\usepackage{amsmath}
\usepackage{bbm, dsfont}
\usepackage{mathrsfs}
\usepackage{pgfplots}    
\pgfplotsset{compat=1.6} 
\begin{document}
\title{Provisioning Time-Based Subscription in NDN: A Secure and Efficient Access Control Scheme}
%
%
\author{ 
Nazatul Sultan\inst{1} \and
Chandan Kumar \inst{2} \and
Saurab Dulal \inst{3} \and
Vijay Varadharajan\inst{4} \and
Seyit Camtepe \inst{1}\and 
Surya Nepal\inst{1}
}
%
%
\institute{
CSIRO's Data61, Australia\and
IIT Kharagpur, India \and
University of Memphis, USA \and
University of Newcastle, Australia  \\
\email{[Nazatul.Sultan; Seyit.Camtepe; Surya.Nepal]@data61.csiro.au}\\
\email{Vijay.Varadharajan@newcastle.edu.au}
}
\maketitle              
\vspace{-.3cm}
\begin{abstract}
This paper proposes a novel encryption-based access control mechanism for Named Data Networking (NDN). The scheme allows data producers to share their content in encrypted form before transmitting it to consumers. The encryption mechanism incorporates time-based subscription access policies directly into the encrypted content, enabling only consumers with valid subscriptions to decrypt it. This makes the scheme well-suited for real-world, subscription-based applications like Netflix. Additionally, the scheme introduces an anonymous and unlinkable signature-based authentication mechanism that empowers edge routers to block bogus content requests at the network's entry point, thereby mitigating Denial of Service (DoS) attacks. A formal security proof demonstrates the scheme’s resistance to Chosen Plaintext Attacks (CPA). Performance analysis, using Mini-NDN-based emulation and a Charm library implementation, further confirms the practicality of the scheme. Moreover, it outperforms closely related works in terms of functionality, security, and communication overhead.

\keywords{Named Data Networking \and Access Control\and Subscription-Based Services\and Revocation\and Encryption\and Provable Security.}
\end{abstract}

\section{Introduction}
\label{intro}
Named Data Networking (NDN) is regarded as a promising future Internet architecture designed to overcome the limitations of traditional IP-based networking architecture \cite{Zhang2014}. The core idea behind NDN is to make data the central focus of network design, rather than the data source. This paradigm shift enables the NDN architecture to address several key issues inherent in IP-based networks, such as bandwidth inefficiencies, security vulnerabilities, scalability challenges, and mobility constraints \cite{Jacobson2006}. Due to these advantages, NDN has garnered considerable attention from researchers in both academia and industry \cite{DJAMA202037}.

\par 
In NDN, every piece of content is identified by a unique \emph{name}. Data packets are cryptographically signed and optionally encrypted by the content provider (or data producer) under this name. The signing process ensures both the integrity and provenance of the data, allowing recipients to verify the content without requiring additional security mechanisms\footnote{In IP-based networks, transport layer security protocols such as SSL and TLS are used to verify content}. Consequently, signed data packets can be distributed and verified independently of their original source, enabling intermediate parties to serve these packets without compromising their security. A consumer sends an \emph{interest packet} for a specific piece of data by specifying its name, and the NDN network forwards this interest toward the data producer. Since NDN routers are equipped with caches, they can respond to the interest if the requested data is stored in their cache; otherwise, the interest is forwarded to the data producer. The data packet then follows the reverse path of the interest packet to reach the consumer \cite{Zhang2014}. As routers forward data packets, they cache copies of them for future requests. This in-network caching significantly improves resource utilization, including network bandwidth, cache space, and content delivery latency, while also enhancing scalability and mobility \cite{Zhang2014}.

\par 
However, while the in-network caching capability of NDN provides many benefits, it also introduces vulnerabilities. Adversaries can compromise routers to access cached data packets. Additionally, since routers can serve content without directly contacting the data producer, unauthorized consumers may obtain content without the producer's knowledge, which can negatively impact business models reliant on controlled access. Therefore, implementing an effective access control mechanism is crucial for the successful deployment of NDN \cite{Tourani2018}. In recent years, various research efforts have addressed the access control problem in data-centric networks, such as \cite{Li2015}, \cite{Fotiou2016}, \cite{Suksomboon2018}, \cite{Bilal2019}, \cite{Misra2019}, \cite{Xia2019}, \cite{Xue2019}, \cite{SultanESORICS}, \cite{SultanSRDS}, \cite{He2021}, \cite{ZHU2020607}. However, most of these existing schemes are account-based, meaning they assume that when a consumer leaves the system, their authorization to access resources terminates. While this is effective in static environments, account-based schemes struggle to support dynamic scenarios where resources are continuously created, and users frequently join and leave the system \cite{Vimercati2012}. For example, account-based approaches are not well-suited for subscription-based access control, particularly time-based subscriptions. Subscription-based access control is a widely used mechanism in modern services. For example, in video streaming platforms like Netflix (\url{www.netflix.com}) and Disney+ (\url{www.disneyplus.com}), data producers offer time-based subscriptions where users can access content for a specific period, such as a week, month, or year. Once the subscription expires, access is revoked. Account-based schemes \cite{Li2015}, \cite{Fotiou2016}, \cite{Suksomboon2018}, \cite{Bilal2019}, \cite{Misra2019}, \cite{Xue2019}, \cite{SultanESORICS}, \cite{SultanSRDS}, \cite{He2021} are inefficient in handling time-based subscriptions due to the frequent need for privilege revocation\footnote{Privilege revocation typically requires updating secret keys for consumers and re-encrypting the data, which is a costly operation \cite{SultanESORICS}.}. 

\par 
To address these challenges, we propose an encryption-based access control mechanism that embeds time-based subscription policies directly into the encrypted content. Only consumers with valid subscriptions associated with these policies can decrypt the content. Our scheme efficiently handles dynamic environments, such as video streaming services like Netflix and Disney+, where authorization is granted for a fixed time period.
\par 
The major contributions of this paper are summarized as follows:
\begin{itemize}
    \item We introduce a novel encryption-based access control scheme, tailored for producers like Netflix and Disney+ to implement time-based subscription policies. This scheme allows producers to offer various subscription options (e.g., day, week, month, half-year, or year) to provide consumers with flexible choices. By embedding time-based access policies directly into the encrypted data, only authorized consumers with a valid subscription can decrypt the content. Our scheme supports multiple subscription types (e.g., day, week, month, year) within the same ciphertext, enabling different consumer groups to select their preferred subscription duration, enhancing scalability and flexibility.
    
    \item Our scheme also presents a signature-based authentication mechanism that enables the edge routers to verify the authenticity of an \emph{interest request} of the consumers before forwarding that request into the network. This is essential to prevent Denial of Service (DoS) Attacks in NDN, which can occur by sending bogus interest requests by malicious entities. Further, the signature mechanism does not reveal the real identity of the requested consumers, which provides consumer anonymity and unlinkability\footnote{The edge routers cannot establish a link between two successive interest requests from the same consumer.}. 
    \item Our scheme also gives formal security proof, which demonstrates its security against Chosen Plaintext Attacks (CPA).
    
    \item We carried out a thorough performance analysis of our scheme by implementing it with the Charm library (\url{www.charm-crypto.io/}) and emulating it using the Mini-NDN tool\footnote{Mini-NDN is a lightweight networking emulation tool based on Mininet (www.github.com/named-data/mini-ndn).}. The results show that our scheme has practical efficiency in terms of computation and communication overhead, making it suitable for real-world NDN applications.
\end{itemize}

The organization of the rest of this paper is as follows: section \ref{related-work} presents a brief overview of closely related works. To increase the clarity of our scheme, we present some preliminaries in Section \ref{preliminaries}. The proposed model (system, adversary, and system models) of our scheme is presented in Section \ref{proposed-model}. In Section \ref{proposed-scheme}, we first present an overview of the proposed scheme followed by its main construction. Section \ref{security-analysis} and Section \ref{performance-analysis} present the security and performance analysis of our scheme, respectively. Finally, Section \ref{conclusion} concludes this paper.

\section{Related Work}
\label{related-work}
In this section, we first present a brief overview of notable encryption-based access control schemes in NDN, followed by relevant related works in key assignment, time-specific encryption, and attribute-based encryption schemes.
\subsection{Access Control Schemes for NDN}
\label{related-work-NDN}
The existing access control schemes for NDN can be divided into two main categories, namely \emph{account-based} and \emph{temporal-based}. Some of the examples of recent account-based schemes are \cite{Li2015}, \cite{Fotiou2016}, \cite{Suksomboon2018},  \cite{Bilal2019},  \cite{Misra2019}, \cite{Xue2019}, \cite{SultanESORICS}, \cite{SultanSRDS}, \cite{He2021}, \cite{He2018}, \cite{Tourani2018a}, \cite{Nunes2018},  \cite{Zhijun2019}, \cite{YANG201921}. We observed that, unlike our scheme, \cite{Bilal2019}, \cite{Tourani2018a}, \cite{Nunes2018}, \cite{Tseng2019} limit the utilization of the NDN resources because the content that is shared for a consumer cannot be accessed by other authorized consumers \cite{SultanSRDS}. On the other hand, in \cite{Zhijun2019}, the consumer needs to contact the producer on every data request. It performs the consumer revocation using a hash table filtering technique, which does not provide privacy. Further, the producer needs to send the updated hash table to each router in the system on every revocation operation, which can increase heavy communication in the system and storage overhead at the routers. Unlike our scheme, there is no consumer interest request verification at the edge router and consumer privacy protection mechanism in \cite{YANG201921}, which can flood the network with bogus interest requests. Further, \cite{YANG201921} also does not provide any type of consumer revocation mechanism. \cite{He2018} does not provide privacy to the consumers, as the signatures reveal the identities of the consumers which eventually also reveals the access pattern of individual consumers. \cite{He2018} also vulnerable to collusion attacks between a revoked consumer and edge routers \cite{SultanESORICS}. The other mentioned schemes like \cite{Li2015}, \cite{Fotiou2016}, \cite{Suksomboon2018}, \cite{Misra2019}, \cite{Xue2019}, \cite{SultanESORICS}, \cite{SultanSRDS}, \cite{He2021} require an expensive (in terms of communication and computation overhead) revocation process. As such, these schemes cannot easily support a highly dynamic environment, where consumers can leave and join the system at any time.
\par 
The temporal-based schemes \cite{Xia2019}, \cite{ZHU2020607} can easily fit into a dynamic environment where consumers' access privileges are given based on some subscription periods. Once their subscription time/period expires, the consumers can automatically revoke from the system. However, \cite{Xia2019} uses a challenge-response protocol for authenticating the consumers at the edge router side for preventing DoS attacks, and the challenge-response protocol increases communication overhead. In \cite{ZHU2020607}, the content producer requires to send the updated proxy re-encryption keys to the content delivery server(s) on each revocation event. Further, it does not provide privacy to the consumers and DoS resistance. Moreover, \cite{Xia2019} and \cite{ZHU2020607} cannot easily embed multiple subscription periods in a single encrypted content. This, in turn, leads to an increase in communication overhead by broadcasting more than one encrypted version of the same content, and it also limits the utilization of NDN resources. 
\vspace{-.2cm}
\subsection{Key Assignment Scheme (KAS)}
\label{related-work:KAS}
The Key Assignment Scheme (KAS) aims to create a key derivation mechanism that allows a provider to encrypt its content using an access key while enabling authorized consumers to derive that access key using their secret keys and corresponding public information. Several key works in this area have been proposed, including notable contributions such as \cite{Atallah2007}, \cite{Crampton2011}, \cite{Vimercati2012}, \cite{Alderman2017}, and \cite{Ferrara2021}. In \cite{Atallah2007}, Atallah \emph{et al.} introduced a KAS method that incorporates access time intervals into the key derivation process. In this approach, only consumers with access rights for a specific time interval can derive the access key necessary for decryption. Crampton, in \cite{Crampton2011}, extended this concept by proposing a KAS that generalizes temporal and geospatial access control policies. This scheme efficiently enforces interval-based access control policies while maintaining low complexity in the enforcement process. Vimercati \emph{et al.} in \cite{Vimercati2012}, presented another KAS scheme allowing a service provider (such as a cloud service) to encrypt content based on subscription to specific time intervals. Building on this, Alderman \emph{et al.} in \cite{Alderman2017} proposed improvements aimed at reducing the amount of public information required in Vimercati’s scheme. In \cite{Ferrara2021}, Ferrara \emph{et al.} introduced a Verifiable Hierarchical Key Assignment Scheme, enabling consumers to verify public information before deriving the access key for decryption.
\par 
However, a common limitation across \cite{Atallah2007}, \cite{Crampton2011}, \cite{Vimercati2012}, \cite{Alderman2017}, and \cite{Ferrara2021} is that these schemes are primarily designed to use the same access key for all content within a specific time interval. For example, if the time interval is set as a day, the same access key is used to encrypt all content for that day. Introducing a different access key for each piece of content within the same time interval would significantly increase key management complexity, as noted in \cite{Zhou2013}.

\subsection{Time Specific Encryption (TSE)}
\label{Related-Works:TRE}
Time-Release Encryption (TRE) allows an encryptor to encrypt a message that cannot be decrypted, even by a legitimate receiver, until a specified release time. Notable works in TRE include \cite{Crescenzo1999}, \cite{Chan2005}, \cite{Cheon2006}, and \cite{Choi2020}. However, these schemes share a common drawback: the reliance on a trusted time server for broadcasting time-specific keys. This setup requires backup mechanisms to ensure that receivers do not miss key broadcasts, adding complexity.
\par 
To address these limitations, the TRE method was later generalized into a new concept called Time-Specific Encryption (TSE), as proposed in \cite{Paterson2010}, \cite{Kasamatsu2012}, and \cite{Ishizaka2020}. In \cite{Paterson2010}, Paterson \emph{et al.} introduced the generalized form of TRE, termed Time-Specific Encryption (TSE). In TSE, an encryptor can define a time interval for decryption, and a semi-trusted time server publishes a global system parameter and periodically issues time instant keys (TIK) that allow receivers to decrypt ciphertexts for specific time periods. This concept was integrated into the \emph{Identity-Based Encryption} (IBE) technique. Later, Kasamatsu \emph{et al.} in \cite{Kasamatsu2012} proposed an improved TSE scheme using forward-secure encryption (FSE) to reduce computation, communication, and storage overhead. In \cite{Ishizaka2020}, Ishizaka \emph{et al.} developed a generic TSE construction based on wildcarded Identity-Based Encryption (WIBE).
\par 
However, these TSE schemes—\cite{Paterson2010}, \cite{Kasamatsu2012}, and \cite{Ishizaka2020}—have a major limitation regarding the efficient utilization of Named Data Networking (NDN) resources. Similar to \cite{Tourani2018a}, \cite{Nunes2018}, \cite{Tseng2019}, and \cite{Bilal2019}, the content shared with one consumer in these schemes cannot be accessed by other authorized consumers, restricting scalability and resource efficiency.
\vspace{-.2cm}

\subsection{Attribute-Based Encryption (ABE)}
\label{related-work:abe}
Attribute-based encryption (ABE) is a public-key encryption technique that allows the encryptor to encrypt a message using attributes such as address, phone number, designation, department, etc. \cite{Sahai2005}. Authorized consumers who possess a certain threshold of attributes can decrypt the message. ABE is a well-explored area; however, most existing ABE schemes do not incorporate time into their access policies and cannot efficiently manage time due to the complexity of assigning and removing attributes dynamically. This limitation is particularly problematic in environments where time intervals change frequently, such as video content sharing in Named Data Networking (NDN), making it computationally and communication-wise expensive for consumers to acquire, remove, or re-grant attributes at different time intervals \cite{Yang2016}. 
\par 
Some notable works that integrate time into ABE are \cite{Zhu2012}, \cite{LIU2014}, \cite{Balani2014}, \cite{Yang2016}, and \cite{Liu2018}. In \cite{Zhu2012}, Zhu \emph{et al.} proposed a temporal ABE scheme utilizing cryptographic integer comparisons and a proxy-based re-encryption mechanism. Similarly, Liu \emph{et al.} in \cite{LIU2014} and Balani \emph{et al.} in \cite{Balani2014} introduced ABE schemes integrating the proxy re-encryption technique. However, the schemes in \cite{Zhu2012}, \cite{LIU2014}, and \cite{Balani2014} require an always-online proxy server to perform the re-encryption operations on consumer requests, which introduces additional overhead. Moreover, these schemes rely on computationally intensive cryptographic operations. In \cite{Yang2016}, Yang \emph{et al.} proposed a time-domain multi-authority ABE scheme that controls access to session keys used for encrypting video content. Time is embedded in both the ciphertext and decryption keys, allowing consumers with the appropriate attributes within a specific time interval to decrypt the content. However, the main drawback of this scheme is that consumers must update their decryption keys for each time interval, further burdening the system. Additionally, \cite{Yang2016} also relies on computationally heavy cryptographic operations. In \cite{Liu2018}, Liu \emph{et al.} addressed the revocation problem by incorporating direct revocation, embedding both the revocation list and authorized time intervals into the ciphertext. However, as the revocation list grows over time, it impacts the size of both the consumer's secret key and the ciphertext, making this scheme unsuitable for large-scale applications such as video streaming services.

	\begin{figure*}[t]
		\centering
		\fbox{\scalebox{5}{\includegraphics[width=1.8cm, height=1.2cm]{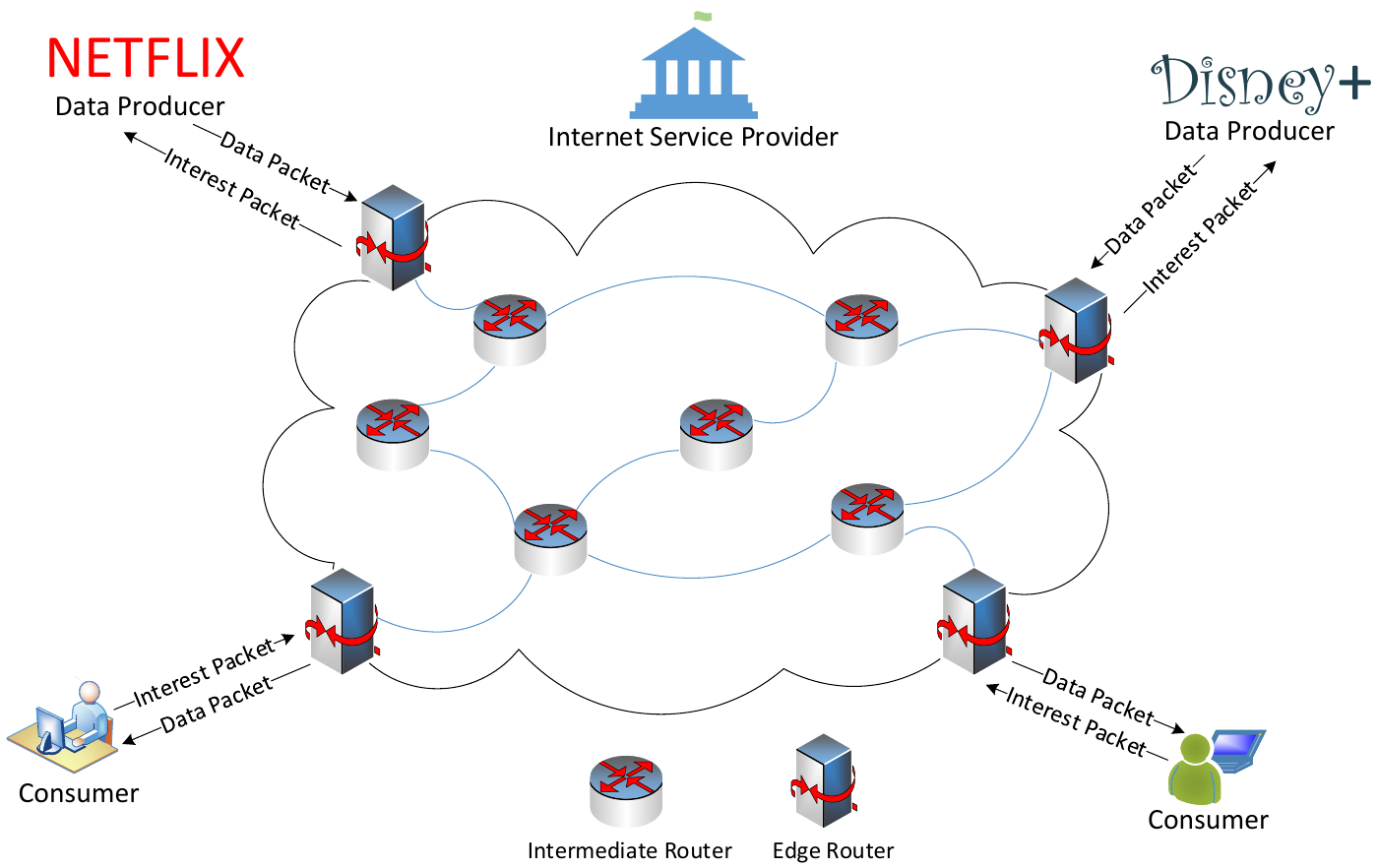}}}
		\caption{Our Proposed NDN Architecture \cite{SultanESORICS}, \cite{SultanSRDS}}
		\label{system_model1}
\vspace{-.3cm}
	\end{figure*}
 \vspace{-.2cm}
\section{Preliminaries}
\label{preliminaries}
\vspace{-.2cm}
In this section, we provide a brief overview of the mathematical concepts and complexity assumptions underlying our proposed scheme.
\vspace{-.2cm}
	\subsection{The Sibling Intractable Function Family (SIFF) \cite{Zheng1993}}
	\label{SIFF}
	SIFF is a technique used to share an encrypted message with a set of authorized users using a polynomial. For illustration, let's assume that we want to share a message, $m$ securely with $n$ authorized users. Each authorized user, $i$ possesses a secret, $\mathtt{sk_i}$. First, the plaintext message $m$ is encrypted using a symmetric key, say $\mathtt{SK}$. Afterward, we choose a polynomial, $P(x)= x^n+ a_1x^{n-1}+ ...+a_{j}x^{n-j}+ ...+ a_n$. The polynomial $P(x)$ can be chosen in such a way that it holds the equation $P(x)= \mathtt{SK}$, and uses $n$ keys of the authorized users as the solutions of the equation to work out all the coefficients. The encrypted message and the coefficients can be published publicly and each authorized user, say $i^{th}$ user can use $P(\mathtt{sk_i})$, for $1\le i\le n$, to compute the secret key $\mathtt{SK}$. This, in turn, enables the user to decrypt the ciphertext using the secret key $\mathtt{SK}$. Further details on SIFF can be found in \cite{Zheng1993}. In this paper, we use this technique to embed different subscription times in the encrypted content itself so that different groups of authorized consumers, subscribing to those subscription times, can decrypt. More details will be given in Section \ref{data-publication}.
	
\vspace{-.2cm}
	\subsection{Bilinear Pairing}
	\label{bilinear_map}
	Bilinear pairing is a function of the modified Weil/Tate pairing defined on the Elliptic Curve. Let $g$ be a random generator of $\mathbb{G}_1$. The Bilinear pairing $\hat{e}:\mathbb{G}_1\times \mathbb{G}_1\rightarrow \mathbb{G}_T$ has the following properties:
	\begin{itemize}
		\item \textit {Bilinear}: $\hat{e}(g^a, g^b)= \hat{e}(g, g)^{ab}$, $\forall g$ and $\forall a, b \in \mathbb{Z}_q^{*}$; where $\mathbb{Z}_q^{*}$ denotes the multiplicative group of $\mathbb{Z}_q$, the integer modulo $q$.
		\item \textit{Non-degenerate}: if $g$ generates $\mathbb{G}_1$, then $\hat{e}(g, g)$ generates $\mathbb{G}_T$.
		\item \textit{Computable}: there exists an efficient algorithm to compute $\hat{e}(g, g)$, for all $g\in \mathbb{G}_1$.
	\end{itemize}

\subsection{Complexity Assumption}
\label{assumption}
Our scheme is based on the following complexity assumption.
\vspace{-.2cm}
\subsubsection{Decisional Bilinear Diffie-Hellman (DBDH) Assumption}
	If $(a, b, c, z)\in \mathbb{Z}_q^*$ are chosen randomly, the ability for any polynomial-time adversary $\mathcal{A}$ to distinguish the tuples $\left<g, g^a, g^b, g^c, Z= \hat{e}(g, g)^{a\cdot b\cdot c}\right>$ and $\left<g, g^a, g^b, g^c, Z= \hat{e}(g, g)^z\right>$ is negligible.
\begin{figure*}[h]
	\centering
	\subfloat[A Sample Time-Based Subscription Access Policy Tree ($\mathcal{T}$)\label{sample-tree}]{\fbox{\scalebox{4.9}{\includegraphics[width=1.4cm, height=.8cm]{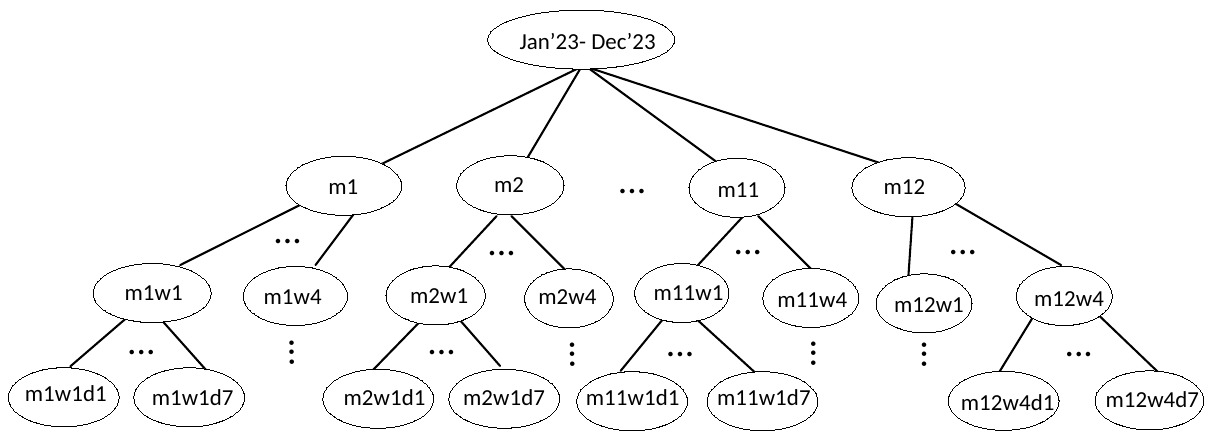}}}}
	~\subfloat[Single Node Illustration\label{node-key-tree}]{\fbox{\scalebox{3.5}{\includegraphics[width=1.2cm, height=1.12cm]{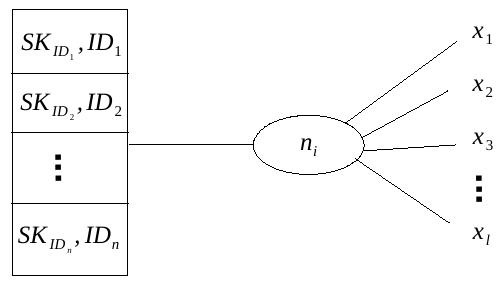}}}}
	\caption{Subscription Access Policy Tree}
 \vspace{-.2cm}
\end{figure*}
\vspace{-.8cm}
	\section{Our Proposed Model}
	\label{proposed-model}
In this section, we present the system, adversary, and security models of our proposed scheme.
\vspace{-.6cm}
	\subsection{System Model}
	\label{system-model}
 \vspace{-.3cm}
	Our system consists of three main entities: \emph{Data Producers}, \emph{Internet Service Providers} (ISPs), and \emph{Consumers}. A typical architecture of the proposed scheme is illustrated in Figure \ref{system_model1}.

    \begin{itemize}
        \item \textbf{Data Producers} are the trusted entities like Netflix, Amazon Prime Video (\url{www.primevideo.com}), and Disney+ that generate content for their subscribed consumers. To access content, a consumer must register with a producer and subscribe by paying a subscription fee. In exchange, the producer provides access rights in the form of private keys, which are issued to the consumer based on the subscription period(s).
        \item \textbf{Internet Service Providers} (ISPs) are responsible for maintaining the NDN network, which is composed of Edge Routers (ERs) and Intermediate Routers (IRs). These routers are equipped with cache memories to store copies of data packets. ERs and IRs forward consumers' interest packets toward their destinations and, when possible, satisfy requests by serving cached data packets, avoiding direct interaction with the producer. ERs also authenticate interest packets before forwarding them to the NDN network, helping to prevent malicious or invalid requests from entering the system.
        \item \textbf{Consumers} are entities that seek to access data from producers. To do so, a consumer must register with a producer and maintain a valid subscription.
    \end{itemize} 

	\subsection{Adversary Model}
	\label{adversary-model}
	Our proposed scheme considers the following adversaries:
\begin{itemize}
    \item \textbf{Malicious Consumers}: Consumers may act maliciously. Two or more revoked consumers, having insufficient access rights, may collude to gain access beyond their actual subscriptions.

    \item \textbf{Honest-but-Curious and Greedy ISPs}: Typically, an ISP is a reputable business organization aiming to increase revenue from producers by providing reliable services. It performs all assigned tasks honestly, creating incentives for producers to use its services. However, the ISP may also be interested in learning as much as possible about the plaintext content, individual consumer identities, and their access patterns. In this context, the ISP is considered an honest but curious entity. Moreover, ISPs can be perceived as greedy, as they might deceive content producers regarding the services provided to consumers for greater financial gain.

    \item \textbf{Malicious Routers}: Routers (i.e., edge and intermediate routers) are managed and maintained by the ISP. However, due to low physical security, an attacker could compromise these routers and leverage the compromised routers to collude with revoked consumers.
\end{itemize}

\vspace{-.5cm}
\subsection{Security Model}
	\label{security-model}

 The security model of our scheme is defined by the Semantic Security against Chosen Plaintext Attacks (IND-CPA). IND-CPA is a security game played between challenger $\mathcal{C}$ and adversary $\mathcal{A}$. We  present the IND-CPA security game in Appendix \ref{security_model_appendix}.

\vspace{-.5cm}
	
	\section{Our Proposed Scheme}
	\label{proposed-scheme}
\vspace{-.2cm}
	We first present an overview of our scheme in Section \ref{overview}, followed by its main construction in Section \ref{main-construction}. 
\vspace{-.4cm}
	 \subsection{Overview}
     \label{overview}
     \vspace{-.1cm}
     Suppose $\mathcal{P}$ is a data producer, and it has $\mathbb{S}$ set of resources. $\mathcal{P}$ allows a consumer to access its resources with a subscription. $\mathcal{P}$ maintains a time-based subscription policy and has several subscription options like day, week, month, and year. A consumer can opt for any one of those subscription options at a time. For example, a consumer can subscribe to the resources for one month from the day of subscription. This type of access control mechanism can be seen in popular video streaming services such as Netflix and DisneyPlus, where the consumers normally pay a certain amount of subscription fees and get access to the videos published by the streaming services for the whole subscribed duration. Note that, although our primary goal in this paper is to design a scheme to support a time-based subscription mechanism, our scheme can also easily be extended to support subscriptions to services defined by different criteria such as the topic of interest, geographical region and hierarchical contents \cite{SultanESORICS}.
    \par 
    To illustrate the time-based subscription policy and the process of granting access to resources for the subscribed consumers, we consider a sample time-based subscription access policy tree, $\mathcal{T}$ as shown in Figure \ref{sample-tree}. For simplicity, in our sample tree, the producer $\mathcal{P}$ maintains subscription options like day, week, month, and year\footnote{Note that our access policy tree can also support more fine-grained time units like hours and seconds, by increasing the height of the tree by one and two more levels, respectively.}. In Figure \ref{sample-tree}, the root node, internal nodes, and leaf nodes represent a particular year, either a month or week and days in that particular year respectively. Suppose a consumer subscribes for the period from $dd/ mm/ yyyy$ to $dd'/ mm'/ yyyy$. It is obvious that this period can be represented by two leaf nodes on the tree, which define the start and end of the subscription respectively. For instance, the subscription period from $07/01/2023$ to $20/08/2023$ can be represented by the leaf nodes $m1w1d7$ and $m8w3d6$ in the tree. The producer $\mathcal{P}$ computes the root nodes of the minimum cover sets in the tree that covers all the leaf nodes between the start and end dates of a subscription. For example, if the subscription period of a consumer is $01/01/2023-31/01/2023$, the minimum cover sets in the tree that covers all the leaf nodes between $m1w1d1$ and $m1w4d7$ will be $\{m1\}$. Again, the root nodes of the minimum cover sets in the tree that covers all the leaf nodes in between $m1w1d2$ and $m8w3d6$ for the subscription period $07/01/2023$ -- $24/08/2023$ are $\{m1w1d7, m1w2, m1w3, m1w4, m2, m3, m4, m5, m6, m7, m8w1, \\m8w2, m8w3d1\}$. After assigning the minimum cover set to the consumer, say $\mathtt{ID_u}$, the producer $\mathcal{P}$ assigns a unique private key set $\mathtt{SK_{ID_u}}$ to the consumer in each node in the minimum cover set. Next, we briefly present a high-level overview of our designed encryption process.
    \par 
    Let us assume $n_i$ is a node in the subscription-based access policy tree $\mathcal{T}$ as shown in Figure \ref{node-key-tree}. Suppose, the producer $\mathcal{P}$ wants to share the content with all the consumers who have the node $n_i$ in their minimum cover sets. The producer $\mathcal{P}$ computes a random value $x_i$, termed as \emph{encryption random secret value} for every $i^{th}$ encryption process of some content $m$ using its master secret and some publicly available parameters. The encryption random secret value is then embedded into the encrypted content $m$ in such a way that anyone who can recompute $x_i$ can recover the access key for decryption. In our designed decryption process, any consumer $\mathtt{ID_u}$ who has the node $n_i$ in his/her minimum cover set can recompute $x_i$ using his/her unique private key $\mathtt{SK_{ID_u}}$, which enables the consumer to recover the access key to obtain the plain content $m$. In our scheme, the producer $\mathcal{P}$ chooses a path from the root to a leaf node in the tree as an access policy and embeds the access policy in the encrypted content. For example, let us say that the producer $\mathcal{P}$'s current date is $7/01/2023$ and wants to share some content, then it will use the path $\{m1w1d7, m1w1, m1, Jan'23- Dec'23\}$ as an access policy for encryption. As such, our encryption scheme needs to be extended from a single node to multiple nodes to satisfy the path $\tau$. Similar to the single node encryption process, the producer $\mathcal{P}$ computes a separate encryption random secret value for every node in the path $\tau$. To embed these secret values into the encrypted content $m$, the producer $\mathcal{P}$ uses the secret values as the roots to compute a polynomial $P(x)$ to share the access key using the SIFF technique \cite{Zheng1993}. This process enables any consumer having one of the nodes in the path $\tau$ in their minimum cover set to recompute the proper encryption random secret value, say $x_i$. Then the consumer can recover the access key from the polynomial $P(x)$ and finally decrypts the encrypted content $m$. Note that the nodes in the chosen path $\tau$ will always have at least one common node in the minimum cover set of all consumers having a valid subscription (i.e., the chosen path belongs to the consumer's subscription period).

  \vspace{-.3cm}
	\subsection{Main Construction}
	\label{main-construction}
 \vspace{-.2cm}
	Our proposed scheme is divided into five phases, namely \emph{Producer Setup}, \emph{Consumer Registration}, \emph{Data Publication}, \emph{Content Request and Consumer Authentication}, and \emph{Decryption}, which are presented next in details.
\vspace{-.4cm}
	\subsubsection{Producer Setup}
	\label{setup}
	The producer $\mathcal{P}$ initiates this phase for the initialization of the system. It generates public parameters $\mathtt{PP}$, master secrets $\mathtt{MS}$, and time-based subscription access policy tree $\mathcal{T}$ for the system. The producer $\mathcal{P}$ keeps the public parameters $\mathtt{PP}$ in its public bulletin board and master secrets $\mathtt{MS}$ in a secure place. The producer $\mathcal{P}$ chooses a large prime number $q$ and two multiplicative bilinear groups $\mathbb{G}_1$ and $\mathbb{G}_T$ of order $q$. It also chooses a bilinear map function $\hat{e}: \mathbb{G}_1\times \mathbb{G}_1\rightarrow \mathbb{G}_T$ and a collision-resistant hash functions $H_1: \{0, 1\}^*\rightarrow \mathbb{Z}^*_q$. The producer $\mathcal{P}$ chooses random numbers $(\sigma, \delta, \kappa, \varkappa)\in \mathbb{Z}_q^*$. It also chooses a random number $\eta_i\in \mathbb{Z}_q^*$ for each node $i$ in the access tree $\mathcal{T}$. Afterward, the producer $\mathcal{P}$ computes, $Y_1= \hat{e}(g, g)^{\kappa}; Y_2= \hat{e}(g, g)^{\varkappa}$. Finally, the producer $\mathcal{P}$ generates the public parameters $\mathtt{PP}$ and master secrets $\mathtt{MS}$ as follows:
	\begin{align*}
	    \mathtt{PP}=& \big<q, \mathbb{G}_1, \mathbb{G}_T, H_1:\{0, 1\}^*\rightarrow \mathbb{Z}^*_q, H_2:\{0, 1\}^*\times \mathbb{Z}_q^*\rightarrow \mathbb{Z}^*_q, \hat{e}, g, Y_1,Y_2\big>\\
	    \mathtt{MS}=& \left<\sigma, \delta, \kappa, \varkappa, \{\eta_i\}_{\forall i\in \mathcal{T}}\right>
	\end{align*}
\vspace{-1cm}
	\subsubsection{Consumer Registration}
    \label{consumer_registration}
  In this phase, the producer $\mathcal{P}$ assigns access rights in the form of private keys to the subscribed consumers. Suppose, a consumer $\mathtt{ID_u}$ wants to subscribe for a time period from $dd- mm- yyyy$ to $dd'- mm'- yyyy'$. The producer $\mathcal{P}$ computes the root nodes of the minimum cover sets in the tree that covers all the leaf nodes of the subscribed period (please refer to Section \ref{overview}). Let the minimum cover set be $\mathbb{CS}$. The producer $\mathcal{P}$ issues a private key set $\mathtt{SK_{ID_u}}= \big<\mathtt{UK_{ID_u}}, \{\mathtt{TK^{1i}_{ID_u}}, \mathtt{TK^{2i}_{ID_u}}\}_{\forall i\in \mathbb{CS}}\big>$ for the whole subscription period, where $\mathtt{UK_{ID_u}}\in \mathbb{Z}_q^*$ a random number and 
    \begin{align*}
        \mathtt{TK^{1i}_{ID_u}}=& g^{\frac{\delta\cdot \mathtt{UK_{ID_u}}}{\eta_i\cdot \eta_i}}\cdot g^{\frac{\sigma\cdot H_1(\mathtt{ID_u})}{\eta_i}}; \mathtt{TK^{2i}_{ID_u}}= g^{\kappa\cdot \mathtt{UK_{ID_u}}\cdot H_1(t_i)}\cdot g^{\varkappa\cdot H_1(\mathtt{ID_u}|| t_i)}
    \end{align*}
Finally, the producer securely sends the private key, $\mathtt{SK_{ID_u}}$ to the consumer. Please note that $t_i$ represents the identity for the $i^{th}$ node in the access tree $\mathcal{T}$.

  \vspace{-.2cm}
    \subsubsection{Data Publication}
    \label{data-publication}
    In this phase, the producer $\mathcal{P}$ publishes its contents to its consumers. The producer $\mathcal{P}$ encrypts the contents before publishing using our proposed encryption mechanism. Only authorized consumers having a proper subscription can decrypt the ciphertext using their private keys. The main idea of our encryption mechanism is to encrypt the actual plaintext content using a secure symmetric-key encryption mechanism with a random secret key, and the secret key is then shared with all the consumers who have a valid subscription. As such, the same encrypted content can be shared with multiple authorized consumers making efficient utilization of NDN resources.
\par 
 The producer $\mathcal{P}$ first encrypts the requested content, say $M$ using a secure symmetric key encryption technique (e.g., AES) with a random symmetric key $\mathtt{K}\in \mathbb{Z}_q^*$ (or access key). The producer $\mathcal{P}$ then selects the leaf node that represents the current date (or publication date) in the time-based subscription access policy tree. The producer $\mathcal{P}$ gets the path, say $\tau$ from the selected leaf node to the root node of the tree and computes the ciphertext components $\big<C_1, \{C_i, x_i\}_{\forall i\in \tau}\big>$ after choosing a random number $r\in \mathtt{Z}_q^*$, where $i$ represents a node in the path $\tau$ and $x_i$ represents the encryption random secret value.
    \begin{align*}
        C_1=& \hat{e}(g, g)^{\sigma\cdot r}; C_i= g^{\eta_i\cdot r}; x_i= \hat{e}(g, g)^{\frac{\delta\cdot r}{\eta_i}}
    \end{align*}
    \par 
    Now the producer $\mathcal{P}$ generates a polynomial $P(x)$ using SIFF scheme (please refer to Section \ref{SIFF}), where $P(x)= x^n+ a_1x^{n-1}+ \cdots +a_ix^{n-i}+\cdots + a_n$.
   For the equation $P(x)= \mathtt{K}$, the producer $\mathcal{P}$ uses all $\{x_i\}_{\forall i\in \tau}$ as the roots to work out all the coefficients denoted as $\mathbb{A}=\{a_1, a_2, \cdots, a_n\}$. Now the producer $\mathcal{P}$ broadcasts the ciphertext $\mathbb{CT}= \left<\mathtt{Enc_{K}}(M), C_1, \{C_i\}_{\forall i\in \tau}, \mathbb{A}\right>$. Please note that the producer $\mathcal{P}$ can choose fresh secret keys to encrypt the actual contents per path $\tau$ in the access policy tree\footnote{We can also enable the producer to choose a fresh secret key to encrypt the actual content and recompute the polynomial $P(x)$ (using the SIFF method) every time the producer receives a fresh content interest request.}. This will reduce the impact of secret key disclosure/compromise. 

    
    \par 
    Our scheme considers a hierarchical content naming method similar to the one defined in \cite{Zhang2014}. For instance, ``\emph{/com/test/$\tau$/abc.mp4/chunk\_1}", where \emph{/com/test/} represents the producer’s domain name (e.g., test.com), \emph{$\tau$} represents the current path in the access tree $\mathcal{T}$, \emph{/abc.mp4} represents a file name, and \emph{/Chunk\_1} specifies the data chunks of the file.
    \vspace{-.2cm}
    \subsubsection{Content Request and User Authentication}
    The main goal of this phase is to verify the authenticity and freshness of a consumer's request so that no bogus interest request can enter the network. This phase has two sub-phases, namely \emph{Content Request} and \emph{Consumer Authentication}, which are described next.
    
 
    \paragraph{Content Request}
    \label{content_request}
    A consumer initiates this phase. In this phase, the consumer generates an interest packet\footnote{The format of the interest packet is similar to the one described in \cite{Zhang2014}} for his/her desired content. The consumer also generates a signature using his/her private key associated with the subscription, requested content name, and current timestamp. Later, the consumer sends both the interest packet and the signature to the nearest edge router. The edge router forwards the interest packet only if it can successfully verify the signature. More detail on the signature verification method is explained in Section \ref{user_auth}. The consumer, say $\mathtt{ID_u}$ generates a signature $\mathtt{Sig_{ID_u}}= \big<S_1, S_2, S_3, S_4, t_i\big>$. The consumer $\mathtt{ID_u}$ chooses a random number $v\in \mathbb{Z}_q^*$ and computes
    
    \begin{align*}
        S_1=& \left[\mathtt{UK_{ID_u}}+ \frac{H_1(ts|| CN)}{H_1(t_i)}\right]\cdot v; S_2=(\mathtt{TK^{2i}_{ID_u}})^v= g^{\left[\kappa\cdot \mathtt{UK_{ID_u}}\cdot H_1(t_i) + \varkappa\cdot H_1(\mathtt{ID_u}|| t_i)\right]v}\\
        S_3=& (Y_1)^{(-v)}= \hat{e}(g, g)^{-\kappa\cdot v}; S_4= (Y_2)^{v\cdot H_1(\mathtt{ID_u}|| t_i)}= \hat{e}(g, g)^{\varkappa\cdot v\cdot H_1(\mathtt{ID_u}|| t_i)}
    \end{align*}
    Here, $CN, ts$ represent the content name and current timestamp. Note that, the consumer uses the current timestamp $ts$ to prevent any replay attack so that the same signature cannot be reused after the predefined time interval. We shall discuss this in the next phase. Also, every signature is randomized using a random number $v$. As such, no entity can link two or more signatures from the same consumer and is able to reveal the real identity of the consumer. 
    \paragraph{Consumer Authentication}
    \label{user_auth}
    After receiving the interest packet and the signature $\mathtt{Sig_{ID_u}}$ from the consumer $\mathtt{ID_u}$, the edge router verifies the signature to check the authenticity of the request in terms of the freshness of the signature, subscription rights and requested content. To verify the freshness of the signature for preventing replay attacks, the edge router uses its current timestamp $ts'$ and checks if the difference between its current timestamp $ts$ and the timestamp ($ts$) received from the consumer is within $\Delta t$, where $\Delta t$ is a predefined value. Otherwise, it aborts the request. As such, no entity can reuse the same signature beyond the threshold value, i.e., $\Delta t$. Next, the edge router verifies the access right of the consumer by using the public information associated with the subscribed time period. As such, this process does not reveal the real identity of the requesting consumers. The edge router verifies the signature and computes $V_1, V_2, V_3$ and $V_4$, where
    \begin{align*}
        V_1=& \hat{e}(S_2, g^{\frac{1}{H_1(t_i)}})= \hat{e}(g^{\left[\kappa\cdot \mathtt{UK_{ID_u}}\cdot H_1(t_i) + \varkappa\cdot H_1(\mathtt{ID_u}|| t_i)\right]v}, g^{\frac{1}{H_1(t_i)}})\\
        =& \hat{e}(g, g)^{\kappa\cdot \mathtt{UK_{ID_u}}\cdot v}\cdot \hat{e}(g, g)^{\frac{\varkappa\cdot H_1(\mathtt{ID_u}|| t_i)\cdot v}{H_1(t_i)}}\\
        V_2=& (Y_1)^{S_1}= \hat{e}(g, g)^{\kappa\cdot \left[\mathtt{UK_{ID_u}}+ \frac{H_1(ts|| CN)}{H_1(t_i)}\right]\cdot v}= \hat{e}(g, g)^{\kappa\cdot \mathtt{UK_{ID_u}}\cdot v}\cdot \hat{e}(g, g)^{\frac{\kappa\cdot  H_1(ts|| CN)\cdot v}{H_1(t_i)}}\\
        V_3=& \frac{V_1}{V_2}= \frac{\hat{e}(g, g)^{\kappa\cdot \mathtt{UK_{ID_u}}\cdot v}\cdot \hat{e}(g, g)^{\frac{\varkappa\cdot H_1(\mathtt{ID_u}|| t_i)\cdot v}{H_1(t_i)}}}{\hat{e}(g, g)^{\kappa\cdot \mathtt{UK_{ID_u}}\cdot v}\cdot \hat{e}(g, g)^{\frac{\kappa\cdot  H_1(ts|| CN)\cdot v}{H_1(t_i)}}}
        = \hat{e}(g, g)^{\frac{\varkappa\cdot H_1(\mathtt{ID_u}|| t_i)\cdot v}{H_1(t_i)}}\cdot \hat{e}(g, g)^{-\frac{\kappa\cdot  H_1(ts|| CN)\cdot v}{H_1(t_i)}}\\
        V_4=& (S_4)^{\frac{1}{H_1(t_i)}}\cdot (S_3)^{\frac{H_1(ts|| CN)}{H_1(t_i)}}= \hat{e}(g, g)^{\frac{\varkappa\cdot v\cdot H_1(\mathtt{ID_u}|| t_i)}{H_1(t_i)}}\cdot \hat{e}(g, g)^{-\frac{\kappa\cdot H_1(ts|| CN)\cdot v}{H_1(t_i)}}
    \end{align*}
    The edge router compares $V_3\stackrel{?}{=} V_4$. If $V_3== V_4$ the consumer is successfully authenticated; otherwise, the edge router drops the interest request. We can see that, the edge router uses the hash of a node's identity ($H_1(t_i)$) in $V_1$ and $V_3$. Similarly, the edge router uses the hash of the concatenation of the timestamp $ts$ and content name $CN$ in $V_4$. As such, if the verification is successful, it indicates that all the information such as timestamp, content name, and subscription rights are legitimate; and the consumer is allowed to access the content. It also enables our scheme to prevent replay attacks by suitably configuring $\Delta t$ to a small value, as per the security best practice. Further, as the signature verification is performed using public parameters unrelated to the actual consumer, the real identity of the consumer $\mathtt{ID_u}$ is not revealed or can be linked to.
    \subsubsection{Decryption}
    \label{decryption}
    Once the interest packet is forwarded by the edge router to the network, eventually the consumer will receive his/her requested content either from the routers' caches or the actual producer. In this phase, the consumer receives the requested ciphertext and accesses the plaintext content after decryption. Let's assume that, the consumer receives the ciphertext, $\mathbb{CT}= \big<\mathtt{Enc_{K}}(M), \\C_1, \{C_i\}_{\forall i\in \tau}, \mathbb{A}\big>$ from the NDN network. The consumer finds a proper private key, $\mathtt{TK^{1i}_{ID_u}}$ for the path $\tau$ associated with the ciphertext $\mathbb{CT}$ and computes the followings:
    
    \begin{align*}
        D^i_1=& \hat{e}(\mathtt{TK^{1i}_{ID_u}}, C_i)= \hat{e}(g^{\frac{\delta\cdot\mathtt{UK_{ID_u}}}{\eta_i\cdot \eta_i}}\cdot g^{\frac{\sigma\cdot H_1(\mathtt{ID_u})}{\eta_i}}, g^{\eta_i\cdot r})= \hat{e}\left(g, g\right)^{\frac{\delta\cdot \mathtt{UK_{ID_u}}\cdot r}{\eta_i}}\cdot \hat{e}(g, g)^{\sigma\cdot H_1(\mathtt{ID_u})\cdot r}\\
        D^i_2=& \frac{D^i_1}{(C_1)^{H_1(\mathtt{ID_u})}}= \frac{\hat{e}\left(g, g\right)^{\frac{\delta\cdot \mathtt{UK_{ID_u}}\cdot r}{\eta_i}}\cdot \hat{e}(g, g)^{\sigma\cdot H_1(\mathtt{ID_u})\cdot r}}{\hat{e}(g, g)^{\sigma\cdot H_1(\mathtt{ID_u})\cdot r}}= \hat{e}\left(g, g\right)^{\frac{\delta\cdot \mathtt{UK_{ID_u}}\cdot r}{\eta_i}}\\
    x'_i=& (D^i_2)^{\frac{1}{\mathtt{UK_{ID_u}}}}= \hat{e}\left(g, g\right)^{\frac{\delta\cdot r}{\eta_i}}
    \end{align*}
    \par 
    Finally, the consumer computes $P(x'_i)$ and recovers the symmetric key $\mathtt{K}$ if and only if he/she has a valid subscription in the form of private keys. The consumer then gets the plaintext content $M$ from $\mathtt{Enc_{K}}(M)$ after decryption.

\begin{remark}
The scheme presented in Section \ref{main-construction} provides a time (or expiration)-based privilege revocation mechanism. However, there may be situations where consumers wish to unsubscribe, or the CP needs to revoke access rights before the subscription’s natural expiration. In such cases, it becomes necessary to immediately revoke the consumer's existing access rights to prevent further access to future content using their current access tokens (private keys). Due to page limitations, we present our immediate privilege revocation mechanism in Appendix \ref{sec:revocation}. 
\end{remark}



\begin{table*}[!t]
\centering
\tiny
\caption{Computation, Communication, and Storage Overhead Comparison}
\label{comp-table}
\begin{tabular}{|c|c|p{3cm}|p{3cm}|p{3cm}|}
\hline
\multicolumn{2}{|l|}{} &  \cite{ZHU2020607} &\cite{Xia2019} & Ours  \\ \hline
\multicolumn{1}{|c|}{\multirow{2}{*}{Req. \& Auth.}} & Req.  &n/a & $2T_p+ T_{de}$ &   $T_{g_1}+ 2T_{g_t}$ \\ \cline{2-5} 
\multicolumn{1}{|c|}{} & Auth. &  n/a & $4T_p+ 3T_{dec}$ & $T_{g_1}+ 3T_{g_t}+ T_p$   \\ \hline
\multicolumn{2}{|l|}{Re-Key Gen} & $4T_p+ 7T_{g_1}$ & n/a & n/a  \\ \hline
\multicolumn{2}{|l|}{Encryption} &  $(5+ 2|\tau''|)T_{g_1}+ T_p+ T_{se}$ & $4T_{g_1}+ 2T_{g_t}+ (|S|+ 2)T_{se}$ & $|\tau|\cdot T_{g_1}+ (|\tau|+ 1)T_{g_t}+ T_{se}$ \\ \hline
\multicolumn{2}{|l|}{Decryption} & $5T_p+ T_{g_1}+ T_{de}$ & $4T_p+ 3T_{de}$ & $2T_{g_t}+ T_p+ T_{de}$   \\ \hline
\multicolumn{2}{|l|}{Ciphertext Size} &  $|\mathbb{G}_T|+ (3+ 2|\tau''|)|\mathbb{G}_1|+ |\mathtt{Enc_K}|$  & $[(4|S|+ 1)|\mathtt{Enc_K}|+ 4|\mathbb{G}_1|]^{\#1}$ &   $|\tau||\mathbb{G}_1|+ |\mathbb{G}_T|+ (|\tau|+ 1)|\mathbb{Z}_q^*|+ |\mathtt{Enc_K}|$\\ \hline
\multicolumn{2}{|l|}{Re-Enc Cipher Size} &  $4|\mathbb{G}_T|+ (4+ |\tau''|)|\mathbb{G}_1|+ (2t+ 1)|\mathbb{Z}_q^*|+ |\mathtt{Enc_K}|$  & n/a &   n/a\\ \hline
\multicolumn{2}{|l|}{Consumer Key Size} & $4|\mathbb{G}|+ 2|\mathbb{Z}_q^*|$  & $2|S||\mathbb{G}_1|+ |S|\mathbb{Z}_q^*$& $2|\mathbb{CS}|\cdot |\mathbb{G}_1|+ |\mathbb{Z}_q^*|$\\ \hline
\multicolumn{2}{|l|}{Signature Size} &  n/a & $(2|\mathbb{G}_1+ 2|\mathbb{Z}_q^*|)^{\#3}$& $|\mathbb{Z}_q^*|+ |\mathbb{G}_1|+ 2|\mathbb{G}_T|$  \\ \hline
\end{tabular}
\\ \raggedright  n/a: not applicable; $|\mathbb{CS}|$: the number of nodes in the minimum cover set $\mathbb{CS}$; $|\tau|$: height of the subscription tree $\mathcal{T}$; $\tau''$: number of access subscription time \cite{ZHU2020607}; $t$: degree of a polynomial \cite{ZHU2020607}; $|S|$: number of samples for challenge-response protocol. $\#1$: $(2|S|+ 1)|\mathtt{Enc_K}|+ 2|\mathbb{G}_1|$ overhead for the challenge samples and $(2|S|+ 1)|\mathtt{Enc_K}|+ 2|\mathbb{G}_1|$ overhead for the time tokens; $\#2$: this is additional storage and communication overhead. This is the overhead incurs for sending consumers' information to the requested ER by the producer for signature verification; $\#3$: here we consider the communication overhead incurred by the challenge-response protocol.
\end{table*}
\section{Security Analysis}
\label{security-analysis}
This section shows that our scheme is secure against Chosen Plaintext Attacks (CPA). 
\vspace{-.2cm}
\subsection{CPA Security}
The CPA security is defined by the following theorem and proof.
\begin{theorem}
\label{cpa}
If a probabilistic polynomial time (PPT) adversary $\mathcal{A}$ can win the CPA security game defined in Section \ref{security-model} with a non-negligible advantage $\epsilon$, then a PPT simulator $\mathcal{B}$ can be constructed to break DBDH assumption with non-negligible advantage $\frac{\epsilon}{2}$.
\end{theorem}
\begin{proof}
In this proof, we consider an adversary $\mathcal{A}$ against our scheme with an advantage $\frac{\epsilon}{2}$. We will construct a PPT simulator $\mathcal{S}$ that will interact with the adversary $\mathcal{A}$ to break our scheme. The detailed proof is presented in Appendix \ref{security_proof_appendix}.
\end{proof}

\begin{table}[t]
\centering
\caption{Storage and Communication Overhead (in bytes)}
\begin{tabular}{|c|c|c|c|c|}
\hline
\multicolumn{2}{|c|}{} & \cite{ZHU2020607} & \cite{Xia2019} & Ours \\ \hline
\multicolumn{2}{|c|}{Ciphertext Size$^*$} & $832$ & $640$ & $484$\\ \hline
 \multicolumn{2}{|c|}{ReEnc Cipher Size$^*$} & $1204$ & n/a & n/a \\ \hline
 \multicolumn{2}{|c|}{Consumer Key Size}  & $296$ & $592$ & $532$\\ \hline
 \multicolumn{2}{|c|}{Signature Size} & n/a & $168$ & $340$\\\hline
\end{tabular}
\label{storage_bytes}
 \vspace{-.5cm}
\end{table}




\section{Performance Analysis}
\label{performance-analysis}
\vspace{-.4cm}
This section analyses the performance of our scheme and compares it with the existing temporal-based Xia \emph{et al.}'s \cite{Xia2019} and Zhu \emph{et al.}'s \cite{ZHU2020607} schemes. We consider only Xia \emph{et al.}'s \cite{Xia2019} and Zhu \emph{et al.}'s \cite{ZHU2020607} schemes, as they are the only schemes that are closely related to ours. Please note that our scheme is better than \cite{Xia2019} and \cite{ZHU2020607} in terms of functionality and security (please refer to Section \ref{related-work}). We start this section by providing a theoretical performance analysis in Section \ref{theory-performance} and then the experimental results in Section \ref{experiments}. 

\vspace{-.4cm}
\subsection{Theoretical Performance Analysis}
\label{theory-performance}
\vspace{-.2cm}
Table \ref{comp-table} shows a comparison between our scheme and \cite{Xia2019}, \cite{ZHU2020607} in terms of computation, communication, and storage overhead. We consider the most frequently operated phases such as \emph{Consumer Request and Authentication} (i.e., \emph{Req. \& Auth.}), \emph{Encryption}, \emph{Re-Key Gen} (which is one of the most frequently operated phases in \cite{ZHU2020607}), and \emph{Decryption} for comparing computation overhead. The computation cost is shown in terms of expensive cryptographic operations, namely exponentiation operation ($T_{g_1}, T_{g_t}$), pairing operation ($T_p$), symmetric key encryption operation ($T_{se}$), and symmetric key decryption operation ($T_{de}$). The computation cost shown in Table \ref{comp-table} represents the asymptotic upper bound in the worst cases. The size of the private key of a consumer, ciphertext, re-encrypted ciphertext (\emph{Re-Enc Cipher}), and signature are considered for comparison of the communication and storage overhead. The comparison is done in terms of group element size (i.e., $|\mathbb{Z}_q^*|, |\mathbb{G}_1|, |\mathbb{G}_T|$) and symmetric key encryption ciphertext size (i.e., $|\mathtt{Enc_K}|$). Table \ref{storage_bytes} presents the storage and communication overhead comparison caused by the ciphertext, consumer key and signatures in terms of bytes. To make the comparison compatible among the schemes, we set the following parameters $|\tau|= |\tau''|= |S|= \mathbb{CS}= t= 4$. Note that we have not considered the overhead caused by the actual encrypted content (i.e., $\mathtt{Enc_{K}}(M)$) in Table \ref{storage_bytes} for measuring the ciphertext size and ReEnc Cipher Size to show the actual overhead difference between our scheme, \cite{Xia2019} and \cite{ZHU2020607} (as $\mathtt{Enc_{K}}(M)$ will be the same to all the three schemes).
\par 
Our scheme takes three group exponentiation operations to generate a signature. It also takes four exponentiation and one pairing operation for the verification of the signature at the edge routers. Note that \cite{ZHU2020607} scheme does not provide the consumer's interest request authentication mechanism. As such, \cite{ZHU2020607} does not provide resistance against DoS attacks like our scheme does. The encryption cost in the \emph{Data Publication} phase of our scheme depends on the height, $|\tau|$ of the subscription tree. We observe that with a height of $4$, our scheme can support the subscription-based access policies for popular real-world video streaming services like Netflix. Table \ref{comp-table} also shows that our scheme takes less encryption cost compared with \cite{ZHU2020607}. On the other hand, the encryption cost in \cite{Xia2019} depends on the number of \emph{challenge-response verification tokens}, computed using a symmetric key encryption algorithm. It takes less computation cost compared to ours due to the less expensive symmetric key encryption operations. However, in \cite{Xia2019}, the producer requires to send samples of the challenge-response tokens of size $(2|S|+ 1)|\mathtt{Enc_K}|+ 2|\mathbb{G}_1|$ to each edge router before the edge routers can process consumers' requests. This increases heavy communication overhead in the system unlike ours. Moreover, unlike our scheme, \cite{ZHU2020607} requires an additional phase, i.e., \emph{Re-Key Gen} to generate re-encryption keys for the content distribution servers. Thus, this phase introduces additional computation overhead in the system. In our scheme, a user requires to perform two group exponentiation, one pairing, and one symmetric key decryption operation in the \emph{Decryption} phase. It can be observed that the decryption cost of our scheme is better than \cite{ZHU2020607} and \cite{Xia2019}. 

\par 

The ciphertext size in our scheme depends on the height of the subscription tree. It can be observed from Table \ref{comp-table} and Table \ref{storage_bytes} that the ciphertext size of our scheme is smaller than \cite{Xia2019}, \cite{ZHU2020607}. As such, our scheme incurs less communication and storage overhead in the system. Further, unlike our scheme, in \cite{ZHU2020607}, the content distribution servers first re-encrypt the ciphertexts before sending them to the requested consumers. The re-encrypted ciphertext size increases linearly with the number of access subscription times and degree of a polynomial associated with a ciphertext, which is larger than the ciphertext size of our scheme (please refer to Table \ref{storage_bytes}). A secret key size, that a consumer maintains, in our scheme depends on the number of nodes in the minimum cover set of the subscription tree based on his/her subscription duration. It can be observed from Table \ref{comp-table} and Table \ref{storage_bytes} that a consumer in \cite{ZHU2020607} requires to maintain fewer secret keys than our scheme. However, as we mentioned earlier, \cite{ZHU2020607} does not support any authentication mechanism for the consumer's interest request to prevent DoS attacks. We can observe from Table \ref{storage_bytes} that the signature size of our scheme is larger than \cite{Xia2019}. However, \cite{Xia2019} is based on a challenge-response protocol that requires the producer to send samples of the challenge-response tokens of size $(2|S|+ 1)|\mathtt{Enc_K}|+ 2|\mathbb{G}_1|$ to each edge router before the edge routers can process consumers' requests. This increases heavy communication overhead in the system unlike ours. 

\pgfplotstableread{
id height time
1   3    0.009403717518
2   5    0.01504255533
3   7    0.01914775372
4   9    0.02269904613

}\compdatas

\begin{figure}[!t]
\centering
    \begin{tikzpicture}
    \pgfplotsset{width=7.5cm,height=2.9cm,compat=1.14}
        \begin{axis} [symbolic x coords={3, 5, 7, 9},
            xtick={3, 5, 7, 9},  
            xticklabel style={text height=2ex}, 
            xlabel={Height},
            ylabel={Time(s)},
            ymin=0, ymax=0.03,
            enlarge x limits=0.04,
            enlarge y limits = .2,
            legend style={at={(0.5,1.1)}, anchor=north, legend columns=-1},
            ] 
        
        \addplot[y, blue, mark=*, smooth] table[x index=1, y index=2] \compdatas;
        \legend{Our Scheme};
        \end{axis}
    \end{tikzpicture}
    \caption{Data publication time for varied height of subscription access tree}
    \label{fig:data_pub_time}
\end{figure}
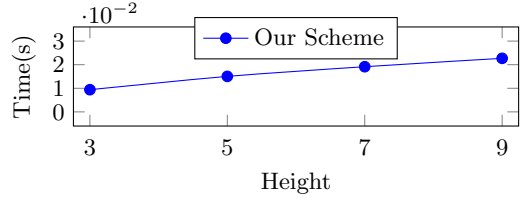
\pgfplotstableread{
id height time
1   1    0.002614784241
2   4    0.00525187254
3   7    0.007980036736
4   11    0.01127227545

}\compdatarev

\begin{figure}[!t]
\centering
    \begin{tikzpicture}
    \pgfplotsset{width=7.5cm,height=2.9cm,compat=1.14}
        \begin{axis} [symbolic x coords={1, 4, 7, 11},
            xtick={1, 4, 7, 11},  
            xticklabel style={text height=2ex}, 
            xlabel={Number of revoked nodes},
            ylabel={Time(s)},
            ymin=0, ymax=0.03,
            enlarge x limits=0.04,
            enlarge y limits = .2,
            legend style={at={(0.5,1.1)}, anchor=north, legend columns=-1},
            ] 
        \addplot[y, blue, mark=*, smooth] table[x index=1, y index=2] \compdatarev;
        


        \legend{Our Scheme};
        \end{axis}
    \end{tikzpicture}
    \caption{Consumer key update time with varied revoked nodes}
    \label{fig:consumer_key_update_time}
\end{figure}
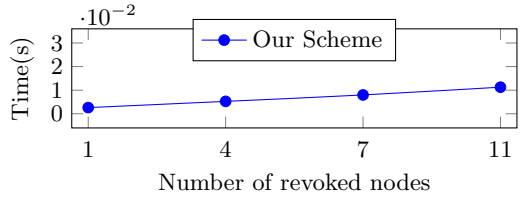

\subsection{Experiments for Cryptographic Algorithms and Network Emulation}
\label{experiments}
\begin{figure*}[!t]
    \centering
    \includegraphics[width=.95\linewidth]{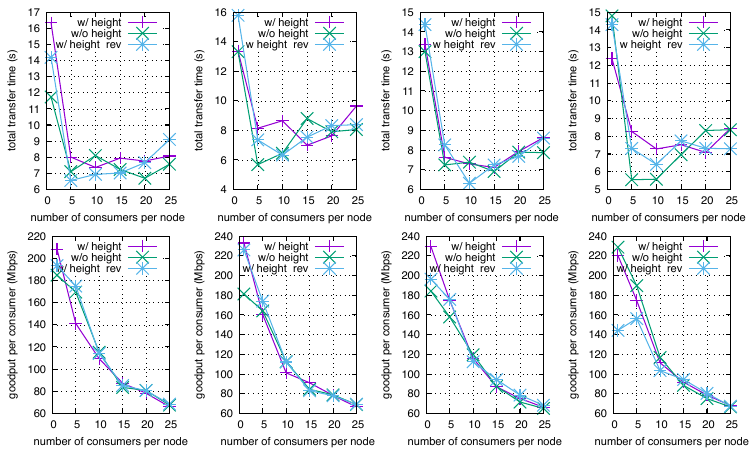}
    \caption{Average total File Transfer Time and Average Goodput with an increasing consumer number}
    \label{fig:transfer-time-goodput}
\end{figure*}
\subsubsection{Experiments for Cryptographic Algorithms:} We implemented our scheme using Charm\footnote{\url{https://github.com/JHUISI/charm}}, Python 3.10.12, and the symmetric group ``SS512" to effectively measure the computational cost of the cryptographic operations. For symmetric key encryption, we utilized AES-256 from the \emph{pycryptodome} library\footnote{\url{https://pypi.org/project/pycryptodome/}}. All experiments were conducted on a standard laptop running Ubuntu 20.04 (64-bit) with a 3GHz Intel Core i5 processor and 20GB of RAM. The reported results are averaged over 20 trials for accuracy.
\par 
Figure~\ref{fig:data_pub_time} shows the data publication time for our scheme. This time is measured by varying the height of the subscription access policy tree. As the tree height increases, the size of the subscription path list grows, leading to a linear increase in data publication time. Figure~\ref{fig:consumer_key_update_time} illustrates the consumer key update time as a function of the number of revoked nodes for our immediate privilege revocation mechanism (please refer to Appendix \ref{sec:revocation}). As expected, the key update time increases linearly with the number of revoked nodes due to the proportional growth of the header size. 
\par 
Additionally, our scheme exhibits constant time performance for key functionalities like content request signing, user authentication, and decryption, with average times of 1ms, 1.7ms, and 0.7ms, respectively, across 20 trials. This efficient performance highlights the practicality and reliability of our approach in real-world applications.


 \begin{figure*}[!t]
    \centering
    \includegraphics[width=.45\linewidth]{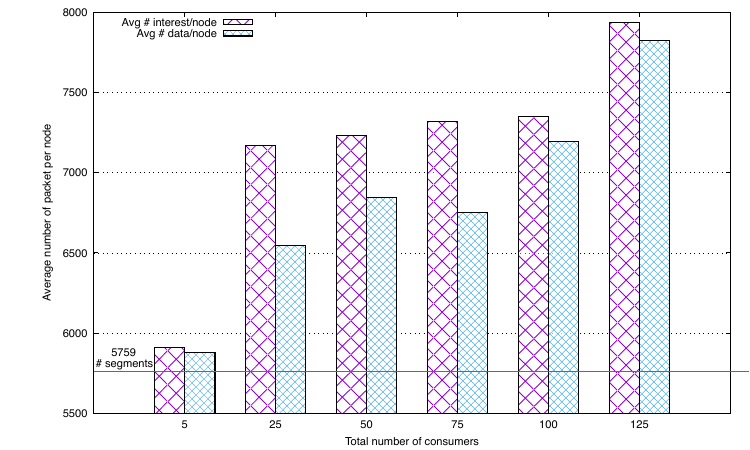}
    \caption{Average Interest and Data count per consumer}
    \label{fig:packet-count}
    \vspace{-.5cm}
\end{figure*}

\begin{figure*}[!t]
    \centering
    \includegraphics[width=.9\linewidth]{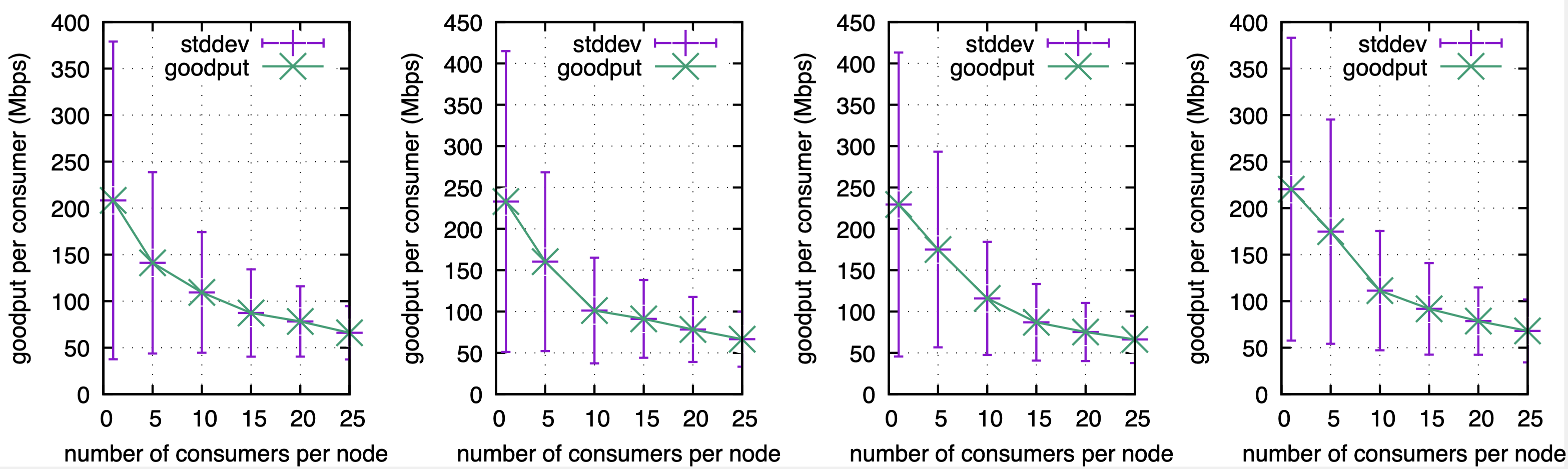}
    \caption{Average Goodput and Standard Deviation per consumer node }
    \label{fig:goodput-std}
    \vspace{-.5cm}
\end{figure*}

		
		


\subsubsection{MiniNDN based Network Emulation}
\label{sec:miniNDN-network-emulation}
We further evaluated the communication overhead of our scheme using Mini-NDN, a lightweight NDN emulator built on Mininet. Using the NDN testbed topology with 37 nodes and 97 links (shown in Figure \ref{topology} in Appendix \ref{Appendix-NDN-topology}), our experiments provided insights into the system's practical performance. We utilized the default link bandwidth and routing protocols provided by Mini-NDN, constrained only by the system's available resources. This setup allowed us to simulate real-world conditions effectively, ensuring that our scheme performs efficiently in a realistic network environment.
\vspace{-.7cm}
\paragraph{Emulation Setup }
Our experiment was designed to thoroughly assess the performance of our proposed scheme, using a producer and a varying number of consumers across three scenarios: (i) the producer published a $46$MB video file as the baseline, (ii) the producer published the video file along with ciphertext files of varying sizes (3, 5, 7, and 9 segments), and (iii) the producer published the video file, ciphertext files, and a revocation file to emulate privilege revocation.
\par 
In the first scenario, we measured the system's baseline performance by transmitting plaintext video content over the NDN network. The second scenario evaluated the system's performance with our proposed scheme applied, while the third scenario assessed the impact of adding the privilege revocation mechanism. The goal was to demonstrate that the performance of our scheme (Scenarios $2$ and $3$) remains close to the baseline, indicating its efficiency even with security mechanisms in place.
\par 
Using the default NDN packet size of $8$Kb, the producer published a total of $5,759$ segments—$5,757$ for the video file, and $1$ segment each for the revocation and height files. Consumers retrieved these segments based on the experiment setup. Furthermore, each consumer transmitted a $340$-byte signature to the producer for verification, included as an Interest parameter, which had no significant impact on performance measurements. 
\par 
To gauge consumer experience, we employed the following key performance metrics:

\begin{itemize}
    \item \textit{File Transfer Time}: The total time it takes for a consumer to fetch a complete file.
    \item \textit{Average Packet Count}: The average number of packets recorded at a node, including transmitted and received Interest or Data packets.
    \item \textit{Goodput}: The ratio of the transmitted file size to the file transfer time, representing efficiency.
\end{itemize}
These metrics provided a comprehensive analysis of the system’s performance, showcasing the robustness and responsiveness of our scheme under various conditions
\vspace{-.5cm}
\subsubsection{Emulation Results} In this section, we present the results of our emulation.
\vspace{-.5cm}
\subsubsection{File Transfer} Figure \ref{fig:transfer-time-goodput} shows the file transfer times across three scenarios: (i) the baseline case, (ii) with height (ciphertext) files only, and (iii) with both height and revocation files. The results demonstrate that there is no significant difference in transfer times across these scenarios, highlighting the efficiency of our scheme in maintaining similar transfer times even when additional security mechanisms are applied. The slight increase in transfer time for a single consumer compared to scenarios with multiple consumers is attributed to the averaging effect. As more consumers join, the first consumer may experience a marginally longer initial retrieval time, but subsequent consumers benefit from caching, leading to significantly faster overall transfer times. Importantly, the increase in consumer count does not result in longer transfer times, thanks to the benefits of aggregation and caching.
\vspace{-.5cm}
\subsubsection{Average Packet Count} Figure \ref{fig:packet-count} depicts the average packet count per consumer node. We evaluated packet counts for different numbers of consumers, with each fetching 5,759 segments (including video, height, and revocation files). For five consumers, we observed a slight increase in the number of Interest and Data packets (5,912 and 5,879 respectively) compared to the expected 5,759 segments, primarily due to application-layer retransmissions for lost Interest packets. However, Interest aggregation and data caching significantly reduced the packet counts as the number of consumers increased. For instance, with 125 consumers, the Interest count was far lower than the expected 719,875 Interests without any optimization. Most importantly, our scheme introduced minimal overhead, adding only two additional segments to the original content, demonstrating its lightweight nature.
\vspace{-.4cm}
\subsubsection{Goodput} Figures \ref{fig:transfer-time-goodput} and \ref{fig:goodput-std} present the average Goodput for each consumer across the three scenarios. Consistent with the file transfer times, all scenarios exhibit similar levels of Goodput, affirming that our scheme does not introduce noticeable performance overhead. While Goodput is generally higher with fewer consumers, it decreases slightly as the number of consumers grows due to the increasing network load. However, this effect is mitigated by local caching, allowing subsequent consumers to achieve higher Goodput after the first consumer fetches content directly from the producer. The standard deviation in Figure \ref{fig:goodput-std} highlights the variability in Goodput among consumers, with later consumers benefiting more from caching.
\vspace{-.4cm}
\section{Conclusion}
\label{conclusion}
\vspace{-.2cm}
In this paper, we proposed an encryption-based access control scheme for NDN that embeds time-based subscription access policies directly into the ciphertext, making it well-suited for real-world applications such as Netflix-like services. Additionally, we designed an anonymous signature-based authentication mechanism that enables edge routers to verify content requests without revealing consumers' identities. This prevents unauthorized or bogus interest requests from entering the network, while also ensuring the unlinkability of consecutive requests from the same consumer. Our scheme is formally proven secure against chosen-plaintext attacks (CPA), and both theoretical and experimental analyses demonstrate its practicality, with efficient computation, communication, and storage overhead compared to related works. Overall, our scheme enhances security and privacy in NDN while maintaining strong performance, making it a promising solution for time-based subscription services.



		\begin{subappendices}\renewcommand{\thesection}{\Alph{section}}
	\section{Security Model}
	\label{security_model_appendix}

 Our scheme consists of five phases, namely \emph{Producer Setup}, \emph{Consumer Registration}, \emph{Data Publication}, \emph{Content Request and Consumer Authentication}, and \emph{Decryption}. The producer initiates the \emph{Producer Setup} phase to generate public parameters and master secrets for the system. In \emph{Consumer Registration} phase, the producer issues access tokens in the form of private keys associated with the subscriptions of the consumers. In \emph{Content Request and Consumer Authentication} phase, the consumer generates interest packets and signatures for the desired contents, and ER verifies the signatures. A consumer initiates the \emph{Decryption} phase to decrypt the received encrypted contents from the producer. Further details are given in Section \ref{proposed-scheme}.\par 
	We define a security model for \emph{Semantic Security against Chosen Plaintext Attacks} (IND-CPA) for our scheme. IND-CPA is illustrated using the following security game between a challenger $\mathcal{C}$ and an adversary $\mathcal{A}$.
	\begin{itemize}
	    \item \textsc{Setup-} Challenger $\mathcal{C}$ runs \emph{Producer Setup} phase to generate public parameters and master secrets. Challenger $\mathcal{C}$ sends the public parameters to the adversary $\mathcal{A}$.
	    \item \textsc{Phase 1-} Adversary $\mathcal{A}$ sends a challenged time period $t_x$. Challenger $\mathcal{C}$ runs \emph{Consumer Registration} phase and generates a private key $\mathtt{SK^{t_x}_{ID_u}}$. Challenger $\mathcal{C}$ sends the private key to the adversary $\mathcal{A}$. An adversary can send queries in polynomial time. 
	    
	    \item \textsc{Challenge-} After \textsc{Phase 1} is over, the adversary $\mathcal{A}$ sends a challenged time period $t^*_i$, and two equal length messages $\mathtt{m_0}$ and $\mathtt{m_1}$ to the challenger $\mathcal{C}$. The challenger $\mathcal{C}$ flips a random coin $\mu\in \{0, 1\}$ and encrypts $\mathtt{m_\mu}$ by initiating \emph{Data Publication} phase. Challenger $\mathcal{C}$ sends the ciphertext of $\mathtt{m_\mu}$ to the adversary $\mathcal{A}$. 
	    \item \textsc{Phase 2-} Same as \textsc{Phase 1}.
	    \item \textsc{Guess-} Adversary $\mathcal{A}$ outputs a guess $\mu'$ of $\mu$. The advantage of wining the game for the adversary $\mathcal{A}$ is $\mathtt{Adv^{IND-CPA}}= |Pr[\mu'= \mu]- \frac{1}{2}|$.
	\end{itemize}
	\begin{definition}
	The proposed scheme is semantically secure against Chosen Plaintext Attack if $\mathtt{Adv^{IND-CPA}}$ is negligible for any polynomial-time adversary $\mathcal{A}$.
	\end{definition}
	\begin{remark}
	 In \textsc{Phase 1}, the adversary $\mathcal{A}$ is also allowed to send queries for signature generation. In our security game, the simulator $\mathcal{B}$ gives all the private keys to the adversary $\mathcal{A}$. As such, the adversary can answer all the signature generation queries by itself. Therefore, we do not include the signature generation oracle (i.e., \emph{Content Request and Consumer Authentication}) in \textsc{Phase 1}. 
	\end{remark}

 \section{CPA Security proof}
 \label{security_proof_appendix}

\begin{proof}
In this proof, we consider an adversary $\mathcal{A}$ against our scheme with an advantage $\frac{\epsilon}{2}$. We will construct a PPT simulator $\mathcal{S}$ that will interact with the adversary $\mathcal{A}$ to break our scheme. 
\par 
The main goal of the adversary $\mathcal{A}$ is to find a valid value of the polynomial $P(x)$ to recover the secret key, $\mathtt{K}$ (please refer to Section \ref{data-publication}). Therefore, we have slightly modified the security game for the original scheme. We will allow the simulator $\mathcal{S}$ to send the polynomial, instead of its exponential function like in the original scheme, as a response to the adversary $\mathcal{A}$. Also, we let the adversary $\mathcal{A}$ to challenge a single time interval at a time for the simplicity of our proof. However, it can easily be extended for multiple time intervals. 
\par 
The DBDH challenger $\mathcal{C}$ sends a tuple $\left<A= g^a, B= g^b, C= g^c, Z\right>$ to the simulator $\mathcal{S}$, where the challenger $\mathcal{C}$ randomly chooses $l\in \{0, 1\}$ and $(a, b, c, z)\in \mathbb{Z}_q^*$. It computes $Z= \hat{e}(g, g)^{abc}$ if $l= 0$; otherwise it computes $Z= \hat{e}(g, g)^z$. Now the simulator $\mathcal{S}$ acts as a challenger in the rest of the security game. At the end of the security game, simulator $\mathcal{S}$ outputs $l$.
\par 
\textbf{\textsc{Setup}} Simulator $\mathcal{S}$ chooses random numbers $(\{\vartheta_i, \varrho_i\}_{\forall i\in \mathcal{T}}, \psi, \phi, \\ \theta, \{\mathtt{sk_{id_i}}\}_{\forall i\in \mathbb{U}})\in \mathbb{Z}_q^*$. Simulator $\mathcal{S}$ computes $\hat{e}(A, B)= \hat{e}(g, g)^{a\cdot b}$. It sets $\{\varphi_i =b\cdot\theta-  \frac{ab\cdot \mathtt{sk_{id_i}}}{\vartheta_i}\}_{\forall i\in \mathcal{T}}$. 

Now, the simulator computes $Y_1= \hat{e}(g, g)^{\psi}; Y_2= \hat{e}(g, g)^{\phi}$.
Simulator $\mathcal{S}$ sends $\big<\mathtt{PP}= \big<q, \mathbb{G}, \mathbb{G}_T, \mathbb{Z}_q^*, \hat{e}, g, H_1, Y_1, Y_2\big>\big>$ to the adversary $\mathcal{A}$.
\par 

\textbf{\textsc{Phase 1}} Adversary $\mathcal{A}$ submits a challenged time period $t_x$ to the simulator $\mathcal{S}$. Simulator $\mathcal{S}$ computes $\mathtt{TK^{1x}_{ID_x}}= g^{\frac{a\cdot b\cdot \mathtt{sk_{id_x}}}{\vartheta_x\cdot \vartheta_x}}\cdot g^{\frac{\varphi_x}{\vartheta_x}}= g^{\frac{b\cdot \theta}{\vartheta_x}}= B^{\frac{\theta}{\vartheta_x}}, \mathtt{TK^{2x}_{ID_x}}= g^{\psi\cdot \mathtt{sk_{id_x}}\cdot H_1(t_x)}\cdot g^{\phi\cdot \varrho_x}$. Simulator $\mathcal{S}$ sends $\big<\mathtt{sk_{id_x}}, \mathtt{TK^{1x}_{ID_x}}, \mathtt{TK^{2x}_{ID_x}}\big>$ to the adversary $\mathcal{A}$.
\par 
\textbf{\textsc{Challenge}} Adversary $\mathcal{A}$ sends two equal length messages $m_0$ and $m_1$ and a challenged time-interval $t^*$ to the simulator $\mathcal{S}$. It flips a random coin $\mu\in \{0, 1\}$ and computes $C_1= \hat{e}(g, g)^{\varphi_{t^*}\cdot c}= \hat{e}(g, g)^{(b\cdot \theta- \frac{ab\cdot \mathtt{sk_{id_u}}}{\vartheta_{t^*}})\cdot c}= \hat{e}(B, C)^{\theta}\cdot Z^{\frac{\mathtt{sk_{id_u}}}{\vartheta_{t^*}}}; C_{t^*}= g^{\vartheta_{t^*}\cdot c}= C^{\vartheta_{t^*}}, x_{t^*}= Z^{\frac{1}{\vartheta_{t^*}}}, P^{t^*}(x)= (x- x_{t^*})$. Finally, the simulator $\mathcal{S}$ sends the ciphertext $\mathbb{CT}= \big<C_1, C_{t^*}, P^*(x)\big>$ to the adversary $\mathcal{A}$. 
\textsc{Guess}-- The adversary $\mathcal{A}$ guesses a bit $\mu'$ and sends to the simulator $\mathcal{S}$. If $\mu'=\mu$ then the adversary $\mathcal{A}$ wins the CPA game; otherwise it fails. If $\mu'= \mu$, simulator $\mathcal{S}$ answers ``DBDH'' in the game (i.e. outputs $l= 0$); otherwise $\mathcal{S}$ answers ``random'' (i.e. outputs $l= 1$).
		\par 
		If $Z= \hat{e}(g, g)^{z}$; then $\mathbb{CT}_\mu$ is completely random from the view of the adversary $\mathcal{A}$. So, the received ciphertext $\mathbb{CT}_\mu$ is not compliant to the game (i.e. invalid ciphertext). Therefore, the adversary $\mathcal{A}$ chooses $\mu'$ randomly. Hence, probability of the adversary $\mathcal{A}$ for outputting $\mu'= \mu$ is $\frac{1}{2}$. 
		\par 
		
		If $Z= \hat{e}(g, g)^{abc}$, then adversary $\mathcal{A}$ receives a valid ciphertext. The adversary $\mathcal{A}$ wins the CPA game with a non-negligible advantage $\epsilon$ (according to Theorem \ref{cpa}).  So, the probability of outputting $\mu'= \mu$ for the adversary $\mathcal{A}$ is $\frac{1}{2}+ \epsilon$, where probability $\epsilon$ is for guessing that the received ciphertext is valid and probability $\frac{1}{2}$ is for guessing whether the valid ciphertext $\mathbb{CT}_\mu$ is related to ${m_0}$ or ${m_1}$.
		\par 
		Therefore, the overall advantage $\mathtt{Adv}$ of the simulator $\mathcal{S}$ is $\frac{1}{2}(\frac{1}{2}+ \epsilon+ \frac{1}{2})- \frac{1}{2}= \frac{\epsilon}{2}$.
\end{proof}

 \section{Immediate Privilege Revocation}
\label{sec:revocation}
The scheme presented in Section \ref{main-construction} establishes a time-based access control framework in which consumers' access privileges are automatically revoked when their subscriptions expire. This model is evident in real-world applications, such as video streaming services like Netflix. However, there may be situations where consumers wish to unsubscribe or content providers (CP) need to revoke access rights before the subscription's natural expiration. In such cases, it becomes necessary to immediately revoke the consumer's existing access rights to prevent further access to future content using their current access tokens (private keys). To incorporate immediate privilege revocation, our scheme introduces specific changes (additions) within the \emph{Producer Setup}, \emph{Consumer Registration}, \emph{Data Publication}, and \emph{Decryption} phases described in Section \ref{main-construction}, following an approach similar to that proposed in \cite{Misra2019}.


\subsubsection{Producer Setup}
CP generates a $t$ degree polynomial $u(x)$ using random coefficients $\{a_0, a_1, \cdots, a_t\}\in \mathbb{Z}^*_q$. It also chooses $\{x_i\in \mathbb{Z}^*_q\}_{0\leq i< (t-1)}$ and computes $\{u(x_i)\}_{0\leq i\leq (t-1)}$. It also chooses a generator $g_p$ of a cyclic group $\mathbb{Z}^*_p$. CP also chooses $n$, where $n< p$ and $p$ is a prime number, to compute $n$ additional points in $u(x)$ for $n$ consumers in the \emph{Consumer Registration} phase. We will present the \emph{Consumer Registration} phase next. CP keeps $\{x_i, u(x_i)\}_{0\leq i\leq (t-1)}$ in $\mathbb{E}$. 

\subsubsection{Consumer Registration}
For each registered consumer $\mathtt{ID_u}$, CP chooses a random number $x_u\in \mathbb{Z}_q^*$, computes $u(x_u)$, and shares the tuple $\big<x_u, u(x_u)\big>$ with the consumer along with the private key $\mathtt{SK_{ID_u}}$ (please refer to \emph{Consumer Registration} phase in Section \ref{consumer_registration}). 

\subsubsection{Data Publication}
During the data publication phase, CP chooses a random number $r\in \mathbb{Z}^*_q$, secret key $k\in \mathbb{Z}^*_p$ and computes $U= k\cdot g_p^{a_0\cdot r}, V= g_p^r$. CP also computes partial Lagrangian coefficients, $\Lambda= \{\lambda'_k =\prod_{0\leq j\not=k <t}\frac{x_j}{x_j- x_k}\}_{0\leq k<t}$. CP Computes $E= \big\{\big<x_i, g_p^{r\cdot u(x_i)}\big>\big\}_{\forall x_i\in \mathbb{E}}$. Finally, it generates an additional header $\mathsf{Hdr}= \big<U, V, E, \Lambda\big>$ which will be sent along with the ciphertext $\mathbb{CT}$ (please refer to Data Publication phase in Section \ref{data-publication}). Please note that CP can also broadcast the header $\mathsf{Hdr}= \big<U, V, E, \Lambda\big>$ to enable the non-revoked consumers to update their private keys.

\subsubsection{Decryption}
After receiving the ciphertext $\mathbb{CT}$ and the header $\mathsf{Hdr}= \big<U, V, E, \Lambda\big>$, the consumer first processes the $\mathsf{Hdr}$ to recover the secret key $k$. 
\par 
Consumer calculates Lagrangian coefficient $\{\lambda_k= \lambda'_k\cdot \frac{x_u}{x_u-x_k}\}_{\forall \lambda'_k\in \Lambda}$. He/she also calculates $\delta_1= \prod_{0\leq k<t}{(g_p^{r\cdot u(x_k)})^{\lambda_k}}$, where $g^{r\cdot u(x_k)}\in E$, Lagrangian coefficient $\lambda_u= \prod_{0\leq k<t}\frac{x_u}{x_u- x_k}$ and $\delta_2= V^{u(x_u)\cdot\lambda_u}= g_p^{r\cdot u(x_u)\cdot\lambda_u}$. The Consumer recovers the secret key, $k= \frac{U}{\delta_1\cdot \delta_2}$, where 
        \begin{align*}
            k=& \frac{U}{\delta_1\cdot \delta_2}= \frac{k\cdot g_p^{a_0\cdot r}}{\prod_{0\leq k<t}{(g_p^{r\cdot u(x_k)})^{\lambda'_k}}\cdot g_p^{r\cdot u(x_u)\cdot\lambda_u}}= \frac{k\cdot g_p^{a_0\cdot r}}{g_p^{r\sum^{k= t-1}_{k= 0} \lambda_k\cdot u(x_k)}\cdot g_p^{r\cdot u(x_u)\cdot\lambda_u}}\\
            =& \frac{k\cdot g_p^{a_0\cdot r}}{g_p^{r\cdot \big(u(x_0)\lambda_{x_0}+ u(x_1)\lambda_{x_1}+ \cdots + u(x_{t-1})\lambda_{x_{t-1}}+ u(x_u)\lambda_{x_u}\big)}}= \frac{k\cdot g_p^{a_0\cdot r}}{g_p^{a_0\cdot r}}
        \end{align*}

\par 
We would like to emphasize that our immediate privilege revocation process occurs infrequently. This is evident from real-world examples, such as video streaming services like Netflix, where only a small number of consumers are generally required to be revoked immediately on certain occasions. 
\section{NDN Testbed Topology}
\label{Appendix-NDN-topology}
For our experiments, we use the NDN testbed topology, illustrated in Figure \ref{topology}, within the mini-NDN emulation framework to create a realistic NDN environment. This topology consists of 37 nodes and 99 links, each with routing costs set by the Named-data Link State Routing Protocol (NLSR). With mini-NDN, we can effectively simulate and analyze NDN-specific behaviors and performance in a controlled and scalable way. This setup enables us to evaluate our scheme's performance in an environment that closely approximates real-world NDN network conditions. 
\begin{figure}[ht]
\centering
		\scalebox{5}{\includegraphics[width=2.5cm, height=2cm]{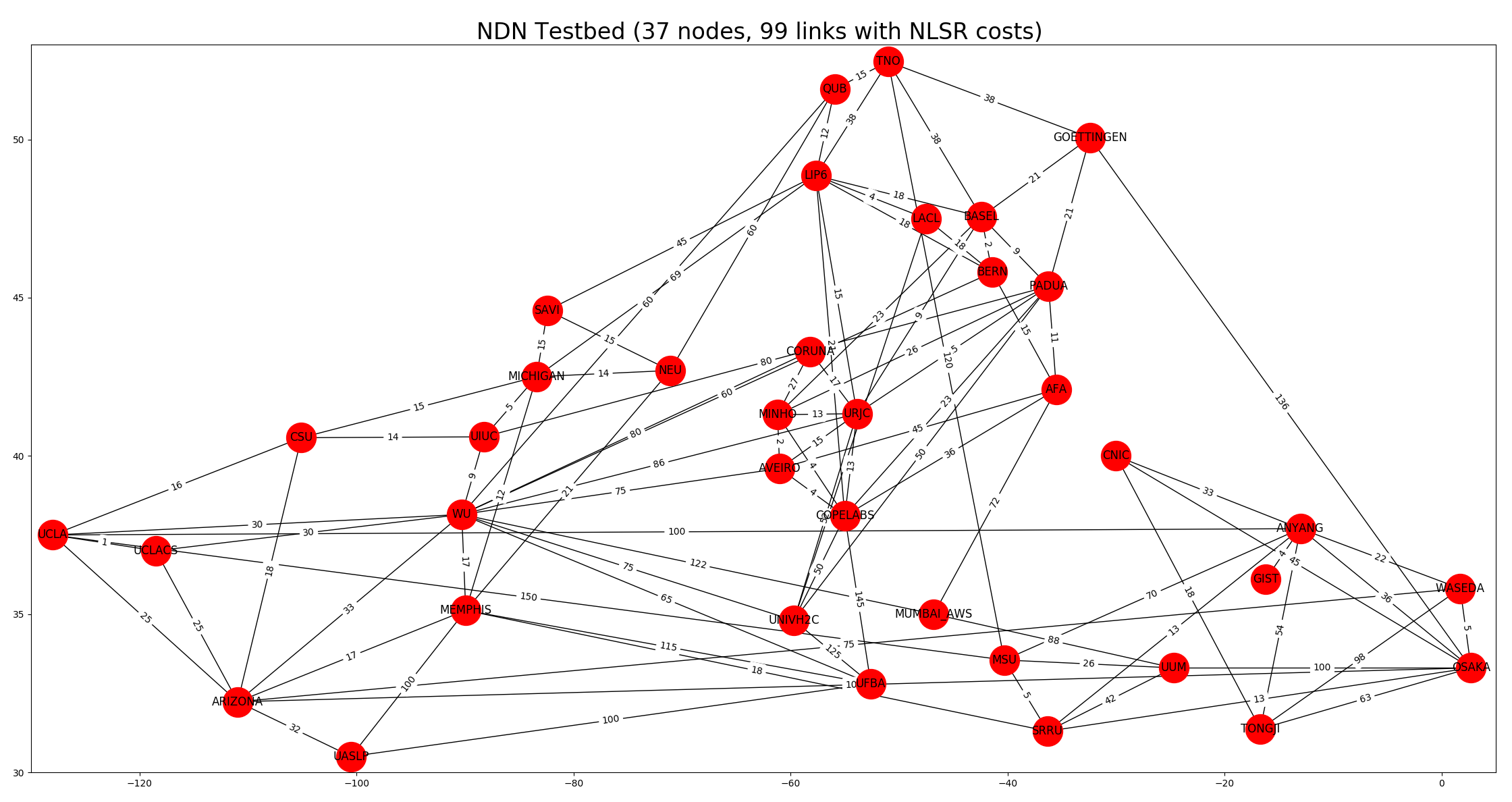}}
		\caption{NDN Testbed \cite{testbed}}
		\label{topology}
		\end{figure}

\end{subappendices}
\end{document}